\documentclass{ifacconf}

\usepackage{graphicx}      
\usepackage{natbib}        
\usepackage{amsmath,amsfonts}
\usepackage{amssymb}

\newtheorem{definition}{Definition}

\newtheorem{lemma}{Lemma}
\begin{document}
\begin{sloppypar}
\begin{frontmatter}

\title{End-to-End ILC for Repetitive Untrackable Tasks: A Cooperative Game Perspective\thanksref{footnoteinfo}} 

\thanks[footnoteinfo]{This work was partially supported by the National Natural Science Foundation of China under Grant 62361136585, partially by the 111 Project under Grant B23008, partially by the National Science Centre Poland under Grant 2023/48/Q/ST7/00205, and partially by the Fundamental Research Funds for the Central Universities under Grant JUSRP202601030.}

\author[JNU]{Zhihe Zhuang} 
\author[TUe]{Rodrigo A. Gonz\'alez}
\author[JNU]{Hongfeng Tao}
\author[UZG]{Wojciech Paszke}
\author[TUe]{Tom Oomen}

\address[JNU]{Key Laboratory of Advanced Process Control for Light Industry (Ministry of Education), Jiangnan University, Wuxi, China. (e-mails: z.h.zhuang@jiangnan.edu.cn;taohongfeng@jiangnan.edu.cn).}
\address[TUe]{Control Systems Technology Section, Department of Mechanical Engineering, Eindhoven University of Technology, Eindhoven, The Netherlands. (e-mails: r.a.gonzalez@tue.nl;t.a.e.oomen@tue.nl).}
\address[UZG]{Institute of Automation, Electronic and Electrical Engineering, University of Zielona G{\'o}ra, Zielona G{\'o}ra, Poland. (e-mail: w.paszke@iee.uz.zgora.pl)}

\begin{abstract}                
An inherent assumption of perfect tracking in iterative learning control (ILC) is that there exists an ILC input such that the generated output can track the desired trajectory reference. This assumption may fail in practice, which gives rise to desired but untrackable tasks. This paper gives an end-to-end ILC design for repetitive untrackable tasks in closed-loop systems. The reference input is trial-to-trial updated together with the ILC feedforward input based on the measurement data. This two-player behavior of the closed-loop ILC system is investigated from a cooperative game perspective. A sufficient condition for the two-player end-to-end ILC to have a lower cost than the one-player norm optimal ILC (NOILC) is discovered. Finally, a numerical example is given to verify the effectiveness of the developed method.
\end{abstract}


\end{frontmatter}

\section{Introduction}
Iterative learning control (ILC) is a useful control strategy for finite-time repetitive tracking tasks where repetitive disturbances can be reduced with enough trials executed, as shown in the early work in \cite{Arimoto1984}. The appeal of ILC lies in its ability to achieve perfect tracking of repetitive tasks over a finite interval. However, the perfect tracking goal usually violates physical restrictions in practice, and thus the desired reference may be untrackable. In \cite{Meng2021TAC_TrackabilityRealizability}, the criteria to check whether a desired reference is trackable to an ILC system is first investigated. Then, the trackable property of ILC has been comprehensively established in \cite{Meng2023JAS_FundamentalTrackability,Wu2023AC_DataBasedTrackability_Meng}. 

To address the ILC problem without the trackability assumption, some research has been conducted to deal with the untrackable repetitive tasks. In \cite{Zhang2023JAS_DataDrivenApproximation_Untrackable}, an input-output-driven ILC design is developed to achieve the best approximation of the untrackable reference with respect to a predefined performance index. The best approximation reference is defined in a mean square sense. In \cite{Wang2024CDC_CompensationForTrackability}, an auxiliary system is learned from the offline input-output test data, which can be used for an ILC compensation strategy to enhance the tracking performance for untrackable tasks.

In practical applications, ILC is usually applied as a feedforward part based on the knowledge from previous trials. When the plant is stabilized first by a feedback controller, the ILC design, acting as a feedforward controller, is referred to as a parallel architecture, which is a common practice as discussed in \cite{Bristow2006survey}. This closed-loop architecture has been widely used in many practical applications, such as precision motion control in \cite{Barton2008TCST_CrossCoupledILC,Oomen2017AC_SparseILC,Zhou2024TII_WaferStageILC} and rehabilitation in \cite{Freeman2016Book_Rehabilitation}. In this case, an untrackable reference would result in the violation of hardware restrictions, which may activate the preset software protection, leading the failure of the current operation.

The solution developed in this paper is to adjust the reference input for the closed-loop systems. From a standard perspective, to fulfill the desired objective with good control performance, the trajectory planning/optimization problem should be carefully addressed first. A stationary trajectory is typically generated for certain restrictions before applying the feedback plus feedforward controllers. In \cite{Lambrechts2005CEP}, a fourth-order trajectory planning algorithm is developed for a feedforward controller with point-to-point moves for high performance. Even if the trajectory is well-designed by using known values of certain restrictions, such as plant inversion information from identification techniques, the inaccuracy of parameters or unknown factors may have a deteriorating influence on the trajectory design, which restricts the further performance improvements of feedforward control. In \cite{TDSon2013AC}, a trajectory learning function is introduced based on the idea of interpolating splines to improve the control performance for point-to-point tracking tasks. In \cite{Cobb2020TCST}, an trial-to-trial path adaptation law is developed to adjust the path shape, aiming to improve the defined performance.


In this paper, the concept of end-to-end, which is standard in machine learning, e.g., self-driving cars \citep{Bojarski2016_EndToEnd} and artificial intelligence training \citep{Silver2017Predictron_EndToEnd}, is introduced to address the untrackable task in the closed-loop ILC systems. In the end-to-end ILC, the trajectory will be trial-to-trial updated based on the measurement data, to both approach the pre-defined ideal trajectory and enhance the robustness against disturbances in practice. A modified norm optimal ILC (NOILC)-like cost function is designed, and explicit update laws for both trajectory and feedforward input updates are derived from the optimization perspective. Note that both reference and feedforward inputs are updated for each trial. The two-player behavior of the closed-loop ILC system is investigated from a cooperative game perspective. The performance improvement of the two-player end-to-end ILC compared to the one-player NOILC is discussed using cooperative game theory. A numerical case study is given to verify the effectiveness of the end-to-end ILC method.


This paper is organized as follows. The problem formulation is first given in Section \ref{section II}. The end-to-end ILC design is presented in Section \ref{section III}. Section \ref{section IV} gives the new result of analyzing the relationship between the reference and ILC input from the cooperative game perspective. A numerical case study is given in Section \ref{section V} to verify the result, and the conclusions are given in Section \ref{section VI}.

Throughout this paper, $\mathbb{N}$ denotes the set of natural number; $\mathbb{R}^{n}$ and $\mathbb{R}^{n \times m}$ denote the sets of $ n $-dimensional real vectors and $ n \times m $ real matrices, respectively. The superscript $ \top $ denotes the transpose operation. $ I_n $ represents the $ n $-dimensional identity matrix and $ \boldsymbol 0 $ denotes the zero vector with compatible dimensions. The notation $A \succeq/\succ 0$ represents that the matrix $A$ is (semi-)positive definite. $\Vert \cdot\Vert_{\boldsymbol Q}$ denotes the weighted Euclidean norm induced by $\boldsymbol Q$.

\section{Problem Formulation}\label{section II}

\subsection{System dynamics}

The closed-loop system in Fig. \ref{fig:ControlDiagram} is considered to execute a class of repetitive tasks, where an ILC implementation combining both parallel and serial architectures is given. Please refer to \cite{Bolder2016AC_InferentialILC} for the design of the serial ILC. The feedback controller $ {\rm C}\left( {{q^{ - 1}}} \right) $ is employed to stabilize the discrete-time linear time-invariant (LTI) plant $ {\rm H}\left( {{q^{ - 1}}} \right) $, where $ q^{-1} $ denotes the discrete-time unit delay. The time duration of each operation is finite with $ N+1 $ time samples. Denote $ t \in \left\lbrace 0, 1, \cdots, N \right\rbrace $ be the time instant and $ k $ be the trial index of a certain repetitive task. Then, the repetitive system dynamics is represented as follows:
\begin{equation}\label{eq:SysDynamics_Sec2}
	\left\{ {\begin{aligned}
			{y_k}\left( t \right) &= {G_{{\rm{cs}}}}\left( {{q^{ - 1}}} \right)r_k\left( t \right) + {G_{{\rm{ps}}}}\left( {{q^{ - 1}}} \right)u_k\left( t \right), \\
			u_k^{\rm{mix}}\left( t \right) &= {\rm C}\left( {{q^{ - 1}}} \right){e_k}\left( t \right) + u_k\left( t \right), \\
			{e_k}\left( t \right) &= r_k\left( t \right) - {y_k}\left( t \right),
	\end{aligned}} \right.
\end{equation}
where $ {y_k}\left( t \right) $, $ {e_k}\left( t \right) $, $ u_k\left( t \right) $, $ {u^{\rm{mix}}_k}\left( t \right) $ and $ r_k\left( t \right) $ denote the output, process tracking error, feedforward ILC input, actuator input and trajectory input signal of time $ t $ on trial $ k $ of the single-input single-output (SISO) systems, respectively. The complementary sensitivity $ G_{{\rm{cs}}} $ and process sensitivity $ G_{{\rm{ps}}} $ are respectively transfer functions from $ r_k $ and $ u_k $ to $ y_k $, i.e.,
\begin{equation}\label{eq:TransferFunctions_Sec2}
	{G_{{\rm{cs}}}} = \frac{{{\rm H}{\rm C}}}{{1 + {\rm H}{\rm C}}},\quad {G_{{\rm{ps}}}} = \frac{{\rm H}}{{1 + {\rm H}{\rm C}}}.
\end{equation}
In addition, the initial condition is assumed to be strictly reset to a constant, i.e., $ y_k\left( t \right) = 0 $ for all $ t \le 0 $, and the initial reference signal satisfies $ r_k\left( 0 \right) = 0 $. Then, the system dynamics is transformed into a lifted description:
\begin{equation}\label{eq:LiftedSysDynamics_Sec2}
	{\boldsymbol y_k} = \boldsymbol G \boldsymbol u_k + \boldsymbol G_r \boldsymbol r_k,
\end{equation}
where $ \boldsymbol G_r $ and $ \boldsymbol G $ are impulse response matrices of $ G_{{\rm{cs}}} $ and $ G_{{\rm{ps}}} $, respectively. When both $ {\rm H} $ and $ {\rm C} $ are linear, time-invariant and $ {\rm C} $ is implementable, $ \boldsymbol G $ is block low triangular and block Toeplitz as discussed in \cite{Gunnarsson2001AC,Mishra2010TCST}. The input vector of trial $ k $, i.e., $ {\boldsymbol u_k} $, is a time-ordered vector composed of stacked input signals, as are the output and process trajectory vectors $\boldsymbol y_k$ and $ \boldsymbol r_k$. In particular, the process tracking error in this paper is defined as $ \boldsymbol e_k = \boldsymbol r_k - \boldsymbol y_k $.

\begin{figure}[tb]
	\centering
	\includegraphics[width=3.5in]{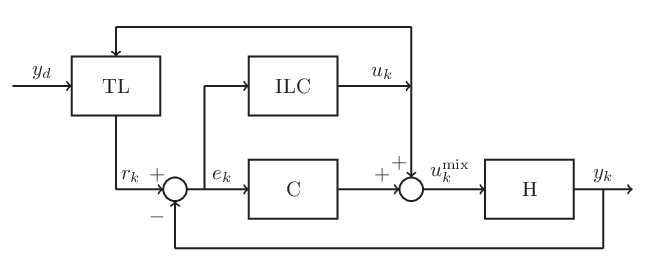}
	\caption{Control block diagram of the end-to-end ILC. \textrm{TL} and \textrm{ILC} are trajectory learning and ILC parts, respectively.}
	\label{fig:ControlDiagram}
\end{figure}

\subsection{Untrackable task}


Define the desired trajectory as $\boldsymbol y_d$. In this paper, we are aiming to track a desired trajectory which is often difficult to track in practical ILC systems, i.e., untrackable task. One example is to complete a tracking task from one point to another as soon as possible. The desired trajectory is a step response, and the controller cannot find a desired input for it with realistic control efforts.

To discuss the untrackable task, the \emph{trackability} property discussed in \cite{Meng2021TAC_TrackabilityRealizability} is introduced. When the desired trajectory is trackable, the process trajectory $ \boldsymbol r_k $ of the closed-loop system (\ref{eq:LiftedSysDynamics_Sec2}) should be the desired trajectory, i.e., $ \boldsymbol r_k = \boldsymbol y_d $. In this case, the ILC system (\ref{eq:LiftedSysDynamics_Sec2}) has a desired input $ \boldsymbol u_d $ satisfying
\begin{equation}\label{eq:DesiredDynamics}
	\boldsymbol y_d = \boldsymbol{\hat G} \boldsymbol u_d,
\end{equation}
where $ \boldsymbol{\hat G} = (\boldsymbol I - \boldsymbol G_r)^{-1}\boldsymbol G $. Then, the \emph{trackability} property is defined as follows.

\begin{definition}[Trackability \citep{Meng2021TAC_TrackabilityRealizability}]\label{def:trackable}
	A desired trajectory $ \boldsymbol y_d $ is trackable for the closed-loop system (\ref{eq:LiftedSysDynamics_Sec2}) if there exists a unique desired input $ \boldsymbol u_d $ fulfilling (\ref{eq:DesiredDynamics}).
\end{definition}

Based on the trackability property in Definition \ref{def:trackable}, this paper tries to give an answer on how to deal with untrackable trajectory references in practice. The idea is to learn a suboptimal trajectory from the experimental data trial by trial, adjusting the trajectory that the closed-loop system follows together with the ILC process. Therefore, the process trajectory $ \boldsymbol r_k $ is introduced and updated trial by trial to approach the desired reference.

To measure the exact tracking performance, the actual error between the desired trajectory and the actual output, which differs from $\boldsymbol e_{k}$, is denoted by $ \boldsymbol {\hat e}_{k} $, where $ \boldsymbol {\hat e}_{k} = \boldsymbol y_{d} - \boldsymbol y_{k} $. Based on the lifted system dynamics (\ref{eq:LiftedSysDynamics_Sec2}), the actual error dynamics is given as
\begin{equation}\label{eq:ErrorDynamics_Sec2}
	\begin{aligned}
		\boldsymbol {\hat e}_{k+1}
		&= \boldsymbol y_{d} - \boldsymbol y_{k+1} \\
		&= \boldsymbol y_{d} - \boldsymbol G \boldsymbol u_{k+1} - \boldsymbol G_r \boldsymbol r_{k+1} \\
		&= \boldsymbol {\hat e}_{k} - \boldsymbol G \Delta\boldsymbol u_{k+1} - \boldsymbol G_r \Delta\boldsymbol r_{k+1},
	\end{aligned}
\end{equation}
where $ \Delta\boldsymbol u_{k+1} = \boldsymbol u_{k+1} - \boldsymbol u_{k} $ and $ \Delta\boldsymbol r_{k+1} = \boldsymbol r_{k+1} - \boldsymbol r_{k} $.

\subsection{Cooperative game perspective}

A new perspective based on cooperative game theory in \cite{Owen2013Book_GameTheory} is introduced to investigate the trial dynamics (\ref{eq:ErrorDynamics_Sec2}) in ILC. The process trajectory $ \boldsymbol r_k $ is introduced to help build the win-win situation together with the ILC input, aiming to achieving a resource-efficient solution to a given untrackable task.

Define a cooperative game with pairs $(\boldsymbol U, v)$ where $\boldsymbol U$ is the set of all players, termed as the grand coalition, and $v$ is the so-called characteristic function that assigns a value to each possible coalition $U \subseteq \boldsymbol U$ of players. Then, $v(U)$ represents the cost of the coalition $U$ when excluding the rest of the players. In this paper, we consider the cooperative game of two players, i.e., the ILC input $\boldsymbol u$ and the process trajectory $\boldsymbol r$, and therefore $\boldsymbol U = \{ \boldsymbol u, \boldsymbol r\}$. The possible coalition $ U $ consists of four cases, i.e., $ U \in \{ \{ \emptyset\}, \{ \boldsymbol u\}, \{ \boldsymbol r\}, \{ \boldsymbol u, \boldsymbol r\} \} $, including the feedback control system only, parallel ILC design, serial ILC design, and the end-to-end ILC developed in this paper.

For simplicity, denote the certain coalition by the corresponding subscript, e.g., $ U_{\boldsymbol u} $, and the characteristic function by omitting the notation $U$, e.g., $v(\boldsymbol u,\boldsymbol r)$. We set $v(\emptyset) = 0$. In particular, the standard cost design of NOILC corresponds to $v(\boldsymbol u,\boldsymbol y_d)$, where the reference signal of the closed-loop system is stationary, i.e., $ \boldsymbol r_{k} = \boldsymbol y_{d} $. We also assume that the cooperative game under consideration is a transferable utility (TU) game \citep{Owen2013Book_GameTheory}, which means the costs can be transferred between different players, i.e., $\boldsymbol u$ and $\boldsymbol r$ in this paper.

\subsection{ILC problem}

In this paper, the objective is to learn the optimal exogenous signals of the ILC input $\boldsymbol u$ and the process trajectory $\boldsymbol r$ for the closed-loop ILC systems. It is equivalent to learning the optimal assignments of the two players with respect to a given characteristic function from the cooperative game perspective. Therefore, the following definition gives the ILC design objective of this paper.

\begin{definition}\label{definition 1}
	The ILC design objective in this paper is to design learning laws for both the ILC input and the process trajectory, i.e.,
	\begin{equation}\label{eq:ILCDesignProblem_Sec2}
		{\left( \boldsymbol u_{k + 1},  \boldsymbol r_{k + 1}\right) } = f\left( \boldsymbol y_d, \boldsymbol e_{k}, \boldsymbol u_{k}, \boldsymbol r_{k} \right),
	\end{equation}
	where $ f(\cdot) $ is a to-be-determined function, to track the given untrackable trajectory reference with lower costs.
\end{definition}

Compared to standard ILC designs, the process trajectory vector $ \boldsymbol r_k $ is introduced to provide a possibility for improvement. Note that the design stated in Definition \ref{definition 1} updates both the ILC input vector $ \boldsymbol u_k $ and the process trajectory vector $ \boldsymbol r_k $ for each trial. This is similar to combining planning and control together, i.e., end-to-end design. In the next section, an end-to-end solution to the ILC problem in Definition \ref{definition 1} is presented.

\section{End-to-End ILC Design}\label{section III}

\subsection{ILC design with trajectory learning}


The problem in Definition \ref{definition 1} can be transformed into the following ILC optimization problem:
\begin{equation}\label{eq:MinOptiProblem_Sec3}
	\begin{aligned}
		&\mathop{\min }\limits \quad J_{k+1}\left(\boldsymbol u_{k + 1},  \boldsymbol r_{k + 1} \right) \\
		&\mathop{\rm s.t.}~~ {\boldsymbol y_k} = \boldsymbol G \boldsymbol u_k + \boldsymbol G_r \boldsymbol r_k.
	\end{aligned}
\end{equation}
In this paper, by drawing on the experience of NOILC, the cost function $ J_{k+1}\left(\boldsymbol u_{k + 1},  \boldsymbol r_{k + 1} \right) $ is designed as follows:
\begin{align}\label{eq:CostUnconstraint_Sec3} 
	J_{k+1}\left(\boldsymbol u_{k + 1},  \boldsymbol r_{k + 1} \right) &= \left\| {\boldsymbol e_{k+1}} \right\|_{\boldsymbol Q}^2 
	\!+\! \left\| {\boldsymbol u_{k+1} \!-\! \boldsymbol u_k} \right\|_{\boldsymbol R}^2 \!+\! \left\| {\boldsymbol u_{k+1}} \right\|_{{\boldsymbol S}}^2 \nonumber \\
	& + \left\| {\boldsymbol y_d \!-\! \boldsymbol r_{k+1}} \right\|_{\boldsymbol W}^2 \!+\! \left\| {\boldsymbol r_{k+1}} \right\|_{{\boldsymbol W}_{\!r}}^2,
\end{align}
where all lifted weighting matrices have a similar form as, e.g., $ \boldsymbol Q = {\rm diag}\left\lbrace Q, Q, \cdots, Q \right\rbrace $, and the weighting parameters $ Q,R,S,W,W_r $ are non-negative.

Note that (\ref{eq:MinOptiProblem_Sec3}) is a multivariable optimization problem. To derive the update laws, let the derivative of the cost function with respect to both $ \boldsymbol r_{k+1} $ and $ \boldsymbol u_{k+1} $ be $ \boldsymbol 0 $, respectively, yielding
\begin{equation}\label{eq:Derivate_rk+1_Sec3}
	\begin{aligned}
		\left[(\boldsymbol I \!-\! \boldsymbol G_r)^{\top} \boldsymbol Q \right.& \left.(\boldsymbol I \!-\! \boldsymbol G_r) \!+ \boldsymbol W \!+\! \boldsymbol W_{\!r} \right]\boldsymbol r_{k+1} \\
		&= (\boldsymbol I - \boldsymbol G_r)^{\top} \boldsymbol Q \boldsymbol G \boldsymbol u_{k+1} + \boldsymbol W \boldsymbol y_d,
	\end{aligned}
\end{equation}
and
\begin{equation}\label{eq:Derivate_uk+1_Sec3}
	\left(\boldsymbol G^{\top}\!\boldsymbol Q\boldsymbol G \!+\! \boldsymbol R \!+\!  \boldsymbol S\right)\!\boldsymbol u_{k+1} \!=\! \boldsymbol R \boldsymbol u_k + \boldsymbol G^{\top}\!\boldsymbol Q (\boldsymbol I - \boldsymbol G_r) \boldsymbol r_{k+1}.
\end{equation}
When the weighting parameters are chosen to be non-zero, $ [(\boldsymbol I \!-\! \boldsymbol G_r)^{\top} \boldsymbol Q (\boldsymbol I \!-\! \boldsymbol G_r)\!+ \boldsymbol W \!+\! \boldsymbol W_{\!r}] $ and $ (\boldsymbol G^{\top}\!\boldsymbol Q\boldsymbol G \!+\! \boldsymbol R \!+\! \boldsymbol S) $ are invertible. Then, substituting (\ref{eq:Derivate_uk+1_Sec3}) to (\ref{eq:Derivate_rk+1_Sec3}) yields
\begin{align}
	\boldsymbol r_{k+1} &=\boldsymbol T_r \boldsymbol u_k + \boldsymbol L_r \boldsymbol y_d, \label{eq:UpdateLaw_rk+1_Sec3} \\
	\boldsymbol u_{k+1} &=\boldsymbol T_u \boldsymbol u_k + \boldsymbol L_u \boldsymbol r_{k+1}, \label{eq:UpdateLaw_uk+1_Sec3} 
\end{align}
where $ \boldsymbol T_r \!=\! \boldsymbol \rho^{-1} (\boldsymbol I \!-\! \boldsymbol G_r)^{\top} \boldsymbol Q \boldsymbol G \boldsymbol T_u $,
$\boldsymbol L_r \!=\! \boldsymbol \rho^{-1} \boldsymbol W $, and
\begin{align}
	\boldsymbol \rho &= (\boldsymbol I \!-\! \boldsymbol G_r )\!^{\top}\![\boldsymbol Q^{-1} \!\!\!+\! \boldsymbol G(\boldsymbol R\!+\!\boldsymbol S)^{-1}\!\boldsymbol G^{\top}\!]^{-1}(\boldsymbol I \!-\! \boldsymbol G_r) \!+\! \boldsymbol W \!\!+\!\! \boldsymbol W_{\!r}, \nonumber \\
	\boldsymbol T_u &= \left(\boldsymbol G^{\top}\boldsymbol Q\boldsymbol G + \boldsymbol R +  \boldsymbol S\right)^{-1} \boldsymbol R, \nonumber \\
	\boldsymbol L_u &= \left(\boldsymbol G^{\top}\boldsymbol Q\boldsymbol G + \boldsymbol R + \boldsymbol S\right)^{-1} \boldsymbol G^{\top}\boldsymbol Q\left(\boldsymbol I \!-\! \boldsymbol G_r \right). \nonumber 
\end{align}

The ILC structure using (\ref{eq:UpdateLaw_rk+1_Sec3}) and (\ref{eq:UpdateLaw_uk+1_Sec3}) is illustrated by the control block diagram in Fig. \ref{fig:ControlDiagram}, which is the combination of parallel and serial ILC. The similar collaborated architecture has been investigated in \cite{deRoover1997thesis}, where the practical implementation instructions are discussed. The aim to introduce this architecture in this paper is to improve the performance for the repetitive untrackable tasks.

\begin{rem}
	There exists a trade-off between the $ \boldsymbol Q $-item $ \Vert {\boldsymbol e_{k+1}} \Vert_{\boldsymbol Q}^2 = \Vert { \boldsymbol r_{k+1} \!-\! \boldsymbol y_{k+1}} \Vert_{\boldsymbol Q}^2 $ and the $ \boldsymbol W $-item $ \Vert {\boldsymbol y_d \!-\! \boldsymbol r_{k+1}} \Vert_{\boldsymbol W}^2 $ in the defined cost function (\ref{eq:CostUnconstraint_Sec3}). If the process trajectory $ \boldsymbol r_{k+1} $ moves too far towards the desired one in a trial, the process tracking error $ \boldsymbol e_{k+1} $ will increase, probably resulting in an increase in the designed cost function (\ref{eq:CostUnconstraint_Sec3}). In this sense, the cost function (\ref{eq:CostUnconstraint_Sec3}) introduces a trajectory learning ability during the ILC process.
\end{rem}


\section{Stability Analysis}\label{section IV}

\subsection{Convergence analysis}

First, the convergence of the control effort is given. Reformulating the update laws (\ref{eq:UpdateLaw_rk+1_Sec3}) and (\ref{eq:UpdateLaw_uk+1_Sec3}) yields
\begin{equation}\label{eq:uk+1&uk_Sec4}
	\begin{aligned}
		\boldsymbol u_{k+1} 
		&=\boldsymbol T_u \boldsymbol u_k + \boldsymbol L_u (\boldsymbol T_r \boldsymbol u_k + \boldsymbol L_r \boldsymbol y_d) \\
		&=(\boldsymbol T_u + \boldsymbol L_u\boldsymbol T_r)\boldsymbol u_k + \boldsymbol L_u\boldsymbol L_r\boldsymbol y_d.
	\end{aligned}
\end{equation}
The sufficient condition for the convergence of control effort in the trial domain is given in the following lemma.
\begin{lemma}\label{lemma_1_Sec4}
	If the norm condition
	\begin{equation}\label{eq:ConvergenceCondition_Sec4}
		\left\| \boldsymbol T_u + \boldsymbol L_u\boldsymbol T_r \right\| < 1,
	\end{equation}
	is satisfied, where matrices in (\ref{eq:ConvergenceCondition_Sec4}) consist of the weighting matrices designed before, the end-to-end ILC design (\ref{eq:UpdateLaw_rk+1_Sec3}) and (\ref{eq:UpdateLaw_uk+1_Sec3}) converges in norm as $k \rightarrow \infty$ and the actual error $ \boldsymbol {\hat e}_{k} $ satisfies
	\begin{equation}
		\boldsymbol {\hat e}_{\infty} = [\boldsymbol I - (\boldsymbol G + \boldsymbol G_r\boldsymbol T_r)(\boldsymbol I - \boldsymbol \xi)^{-1}\boldsymbol L_u\boldsymbol L_r - \boldsymbol G_r\boldsymbol L_r]\boldsymbol y_d,
	\end{equation}
	where $\boldsymbol \xi = \boldsymbol T_u + \boldsymbol L_u\boldsymbol T_r$.
\end{lemma}

\begin{pf}
	It follows from (\ref{eq:uk+1&uk_Sec4}) that
	\begin{equation}\label{eq:uk&u0_Lemma1_Sec4}
		\boldsymbol u_{k} = \boldsymbol \xi^{k}\boldsymbol u_{0} + (\boldsymbol I - \boldsymbol \xi)^{-1}(\boldsymbol I - \boldsymbol \xi^{k})\boldsymbol L_u\boldsymbol L_r\boldsymbol y_d,
	\end{equation}
	which means the ILC control effort $\boldsymbol u_{k}$ converges if (\ref{eq:ConvergenceCondition_Sec4}) is satisfied. When $k \rightarrow \infty$, $\boldsymbol u_{\infty} = (\boldsymbol I - \boldsymbol \xi)^{-1}\boldsymbol L_u\boldsymbol L_r\boldsymbol y_d$. Then, it follows from (\ref{eq:UpdateLaw_rk+1_Sec3}) that the process trajectory $\boldsymbol r_{k}$ converges and $\boldsymbol r_{\infty} = [\boldsymbol T_r(\boldsymbol I - \boldsymbol \xi)^{-1}\boldsymbol L_u + \boldsymbol I]\boldsymbol L_r\boldsymbol y_d$. Then, the steady error $\boldsymbol {\hat e}_{\infty}$ is
	\begin{equation}
		\begin{aligned}
			\boldsymbol {\hat e}_{\infty} 
			&= \boldsymbol y_{d} - \boldsymbol y_{\infty} \\
			&= \boldsymbol y_{d} - \boldsymbol G\boldsymbol u_{\infty} - \boldsymbol G_r\boldsymbol r_{\infty} \\
			&= [\boldsymbol I - (\boldsymbol G + \boldsymbol G_r\boldsymbol T_r)(\boldsymbol I - \boldsymbol \xi)^{-1}\boldsymbol L_u\boldsymbol L_r - \boldsymbol G_r\boldsymbol L_r]\boldsymbol y_d,
		\end{aligned}
	\end{equation}
	which completes the proof.
\end{pf}

\subsection{Cooperative game analysis}

Given an untrackable $\boldsymbol y_d$ in Fig. \ref{fig:ControlDiagram}, there could be several suboptimal solutions with respect to a given cost function. These solutions reflect different coordinations between the two players $\boldsymbol r_k$ and $\boldsymbol u_k$, depending on how the process trajectory $\boldsymbol r_{k}$ is and how the cost function penalizes the error and control energy. For instance, when the two players learn to approach a step-response-shaped task, it is significant to choose a solution that balances the energy consumption and accuracy requirement.

From the cooperative game perspective, the coalition exists to achieve a better performance. If one of the two players can better solve the tracking problem alone, there is no need to build the coalition, which means the coalition is internally unstable. The internal stability is introduced in the following definition.

\begin{definition}[Internal stability of a coalition]\label{def:InternalStability_sec4}
	For $\forall U \subseteq \boldsymbol U$, $U$ is internally stable if $\nexists i \in U$ such that $v(U\backslash\{i\}) > v(U)$.
\end{definition}

Here, the corresponding external stability means that no extra players can have a better payoff if joining the coalition. The external stability of the grand coalition $\boldsymbol U$ is vacuously true since there are no other players in the considered two-player game. Therefore, if $\boldsymbol U$ is internally stable, $\boldsymbol U$ is a stable set of the cooperative game $(\boldsymbol U,v)$ as investigated in \cite{vonNeumann1947Book_GameTheory}. The internal stability demands a better characteristic value of a coalition when including a player, which gives the necessity of forming a coalition.

To show the necessity of forming the two-player coalition, the definition of convex game is given as follows.

\begin{definition}[Convex cooperative game \citep{Shapley1971_ConvexCoopGame}]
	The cooperative game $(\boldsymbol U,v)$ is convex if $\forall U,V \subseteq \boldsymbol U$, it is satisfied that $v(U) + v(V) \le v(U\cup V) + v(U \cap V)$.
\end{definition}

Considering the two-player case, the convex game collapses to a superadditive game where $v(U) + v(V) \le v(U\cup V)$. Therefore, if this superadditive condition is satisfied, the grand coalition $\boldsymbol U$ of the considered two-player game is internally stable according to Definition \ref{def:InternalStability_sec4} and hence a stable set. It is also the unique core of the two-player convex game.

To analyze the developed end-to-end ILC, introduce the cost of NOILC in the given cost function (\ref{eq:CostUnconstraint_Sec3}) as the baseline cost $V^0_{k}$, i.e., $V^0_{k} = J_{k}(\boldsymbol u_{k},\boldsymbol y_d)$. Define the characteristic function of the cooperative game $(\boldsymbol U, v)$ on trial $k$ as
\begin{equation}\label{eq:CharaFunc_Sec4}
	v_{k}(U) = V^0_{k} - J_{k}(U).
\end{equation}

Then, the following theorem finds a simple condition such that the end-to-end ILC design (\ref{eq:UpdateLaw_rk+1_Sec3}) and (\ref{eq:UpdateLaw_uk+1_Sec3}) can make the closed-loop tracking problem in Fig. \ref{fig:ControlDiagram} a convex cooperative game. In other words, the necessity of introducing the two-player end-to-end ILC design is demonstrated.

\begin{thm}\label{theorem_0_Sec4}
	Given symmetric weighting matrices $\boldsymbol Q, \boldsymbol W \succ 0$ and $\boldsymbol S,\boldsymbol W_{\!r} \succeq 0$ satisfying (\ref{eq:ConvergenceCondition_Sec4}) in Lemma \ref{lemma_1_Sec4} and initial ILC input $\boldsymbol u_0 =\boldsymbol 0$, the coalition between the ILC feedforward input $\boldsymbol u_k$ and the process trajectory $\boldsymbol r_k$ is a stable set of the cooperative game $(\boldsymbol U,v)$ if
	\begin{equation}\label{Theorem0_Condition2_Sec4}
		2 \boldsymbol L_u\boldsymbol T_r + \boldsymbol T_u \succeq \boldsymbol I,
	\end{equation}
	with symmetric $\boldsymbol R$.
\end{thm}

\begin{pf}
	For the two-player cooperative game $(\boldsymbol U,v)$, a coalition being a stable set is equivalent to the coalition being internally stable. Then, it is equivalent to prove that
	\begin{equation}\label{eq:ProofTheorem0_0_Sec4}
		v(\boldsymbol u_{k+1}) + v(\boldsymbol r_{k+1}) \le v(\boldsymbol u_{k+1},\boldsymbol r_{k+1}).
	\end{equation}
	In the closed-loop ILC system in Fig. \ref{fig:ControlDiagram}, $\boldsymbol r_{k+1} = \boldsymbol y_d$ when the player $\boldsymbol r_{k+1}$ is not involved and $\boldsymbol u_{k+1} = \boldsymbol 0$ means $U = \{\boldsymbol r\}$. Then, employing (\ref{eq:CharaFunc_Sec4}) yields $v(\boldsymbol u_{k+1}) = V^0_{k}$ and
	\begin{equation}
		v(\boldsymbol r_{k+1}) = J_{k}(\boldsymbol u_{k+1},\boldsymbol y_d) - J_{k}(\boldsymbol 0,\boldsymbol r_{k+1}).
	\end{equation}
	Then, it follows from (\ref{eq:CostUnconstraint_Sec3}) that
	\begin{equation}
		\begin{aligned}
			\Vert \boldsymbol r_{k+1} - \boldsymbol y_{k+1} \Vert_{\boldsymbol Q}^2 + \Vert \boldsymbol u_{k+1} - &\boldsymbol u_k \Vert_{\boldsymbol R}^2 + \Vert \boldsymbol u_{k+1} \Vert_{{\boldsymbol S}}^2 \\
			&\le \Vert \boldsymbol r_{k+1} - \boldsymbol G_r\boldsymbol r_{k+1} \Vert_{{\boldsymbol Q}}^2,
		\end{aligned}
	\end{equation}
	and hence it suffices to prove that
	\begin{equation}\label{eq:ProofTheorem0_1_Sec4}
		\begin{aligned}
			\!\!\!\!\!\Vert\boldsymbol y_{k+1} \Vert_{\boldsymbol Q}^2 - 2\boldsymbol r^{\top}_{k+1}&\boldsymbol Q\boldsymbol y_{k+1} + \Vert {\boldsymbol u_{k+1} - \boldsymbol u_k} \Vert_{\boldsymbol R}^2 + \Vert {\boldsymbol u_{k+1}} \Vert_{{\boldsymbol S}}^2 \\
			&\le \Vert \boldsymbol G_r\boldsymbol r_{k+1} \Vert_{{\boldsymbol Q}}^2 - 2\boldsymbol r^{\top}_{k+1}\boldsymbol Q\boldsymbol G_r\boldsymbol r_{k+1}.
		\end{aligned}
	\end{equation}
	Substituting the system dynamics (\ref{eq:LiftedSysDynamics_Sec2}) into (\ref{eq:ProofTheorem0_1_Sec4}) gives
	\begin{equation}
		\begin{aligned}
			\Vert\boldsymbol G\boldsymbol u_{k+1} \Vert_{\boldsymbol Q}^2 + \Vert \boldsymbol u_{k+1} - &\boldsymbol u_k \Vert_{\boldsymbol R}^2 + \Vert \boldsymbol u_{k+1} \Vert_{{\boldsymbol S}}^2 \\
			&\le 2\boldsymbol u^{\top}_{k+1}\boldsymbol G^{\top}\boldsymbol Q(\boldsymbol I - \boldsymbol G_r)\boldsymbol r_{k+1},
		\end{aligned}
	\end{equation}
	which yields
	\begin{equation}
		\begin{aligned}
			\boldsymbol u^{\top}_{k+1}(\boldsymbol G^{\top}\boldsymbol Q&\boldsymbol G + \boldsymbol R + \boldsymbol S)
			\boldsymbol u_{k+1} - 2\boldsymbol u^{\top}_{k}\boldsymbol R\boldsymbol u_{k+1} \\
			&+ \boldsymbol u^{\top}_{k}\boldsymbol R\boldsymbol u_{k} \le 2\boldsymbol u^{\top}_{k+1}\boldsymbol G^{\top}\boldsymbol Q(\boldsymbol I - \boldsymbol G_r)\boldsymbol r_{k+1},
		\end{aligned}
	\end{equation}
	i.e.,
	\begin{equation}
			\boldsymbol u^{\top}_{k+1}\boldsymbol G^{\top}\boldsymbol Q(\boldsymbol I - \boldsymbol G_r)\boldsymbol r_{k+1} + \boldsymbol u^{\top}_{k}\boldsymbol R\boldsymbol u_{k+1} - \boldsymbol u^{\top}_{k}\boldsymbol R\boldsymbol u_{k} \ge 0.
	\end{equation}
	Then, eliminating $\boldsymbol u_{k+1}$ via (\ref{eq:UpdateLaw_uk+1_Sec3}) gives rise to
	\begin{equation}\label{eq:ProofTheorem1_3.1_Sec4}
		\begin{aligned}
			&\boldsymbol u^{\top}_{k}[\boldsymbol T^{\top}_u\boldsymbol G^{\top}\boldsymbol Q(\boldsymbol I - \boldsymbol G_r) + \boldsymbol R \boldsymbol L_u] \boldsymbol r_{k+1} \\
			&\!\!+ \boldsymbol u^{\top}_{k}(\boldsymbol R\boldsymbol T_u \!-\! \boldsymbol R)\boldsymbol u_{k} \!+\! \boldsymbol r^{\top}_{k+1}\boldsymbol L^{\top}_u\boldsymbol G^{\top}\boldsymbol Q(\boldsymbol I \!-\! \boldsymbol G_r)\boldsymbol r_{k+1} \ge 0.
		\end{aligned}
	\end{equation}
	Since $(\boldsymbol G^{\top}\boldsymbol Q\boldsymbol G + \boldsymbol R +  \boldsymbol S)^{-1}$ is symmetric, and thus $\boldsymbol T^{\top}_u\boldsymbol G^{\top}\boldsymbol Q(\boldsymbol I - \boldsymbol G_r) = \boldsymbol R \boldsymbol L_u$, then (\ref{eq:ProofTheorem1_3.1_Sec4}) can be rewritten as
	\begin{equation}
		\begin{aligned}
			&2\boldsymbol u^{\top}_{k}\boldsymbol R \boldsymbol L_u \boldsymbol L_r\boldsymbol y_d + \boldsymbol u^{\top}_{k}(2\boldsymbol R \boldsymbol L_u\boldsymbol T_r + \boldsymbol R\boldsymbol T_u \!-\! \boldsymbol R)\boldsymbol u_{k} \\
			&\hspace{2.5cm} +\boldsymbol r^{\top}_{k+1}\boldsymbol L^{\top}_u\boldsymbol G^{\top}\boldsymbol Q(\boldsymbol I \!-\! \boldsymbol G_r)\boldsymbol r_{k+1} \ge 0.
		\end{aligned}
	\end{equation}
	Note that
	\begin{align}
		\boldsymbol L^{\top}_u\boldsymbol G^{\top}&\boldsymbol Q(\boldsymbol I \!-\! \boldsymbol G_r) = (\boldsymbol I \!-\! \boldsymbol G_r)^{\top}\boldsymbol Q\boldsymbol G \times \notag \\
		&(\boldsymbol G^{\top}\boldsymbol Q\boldsymbol G + \boldsymbol R + \boldsymbol S)^{-1}\boldsymbol G^{\top}\boldsymbol Q(\boldsymbol I \!-\! \boldsymbol G_r) \succ 0,
	\end{align}
	and $2\boldsymbol R \boldsymbol L_u\boldsymbol T_r + \boldsymbol R\boldsymbol T_u \!-\! \boldsymbol R \succeq 0$ as given in condition (\ref{Theorem0_Condition2_Sec4}), it suffices to prove
	\begin{equation}\label{eq:ProofTheorem1_6_Sec4}
		\boldsymbol u^{\top}_{k} \boldsymbol R \boldsymbol L_u\boldsymbol L_r \boldsymbol y_d \ge 0.
	\end{equation}
	Eliminating $\boldsymbol u_{k}$ via (\ref{eq:uk&u0_Lemma1_Sec4}) in Lemma \ref{lemma_1_Sec4} gives rise to
	\begin{equation}
		\begin{aligned}
			\boldsymbol u^{\top}_{k}\boldsymbol R &\boldsymbol L_u\boldsymbol L_r \boldsymbol y_d = \boldsymbol u^{\top}_0(\boldsymbol \xi^{k})^{\top}\boldsymbol R\boldsymbol L_u\boldsymbol L_r\boldsymbol y_d  \\ 
			&+\boldsymbol y^{\top}_d\boldsymbol L^{\top}_r\boldsymbol L^{\top}_u[(\boldsymbol I - \boldsymbol \xi)^{-1}(\boldsymbol I - \boldsymbol \xi^{k})]^{\top}\boldsymbol R\boldsymbol L_u\boldsymbol L_r \boldsymbol y_d,
		\end{aligned}
	\end{equation}
	where the ILC input of the first trial is chosen as $\boldsymbol u_0 = \boldsymbol 0 $. Note that
	\begin{equation}
		\begin{aligned}
			\boldsymbol \xi 
			&= \boldsymbol T_u + \boldsymbol L_u\boldsymbol T_r \\
			&= (\boldsymbol G^{\top}\boldsymbol Q\boldsymbol G + \boldsymbol R +  \boldsymbol S)^{-1} \boldsymbol R + \boldsymbol L_u \boldsymbol \rho^{-1}\boldsymbol L^{\top}_u \boldsymbol R \succ 0,
		\end{aligned}
	\end{equation}
	and $\Vert\boldsymbol \xi \Vert < 1$, then it follows that $(\boldsymbol I - \boldsymbol \xi) \succ 0$ and $(\boldsymbol I - \boldsymbol \xi^{k}) \succ 0$. Therefore, (\ref{eq:ProofTheorem1_6_Sec4}) holds and the proof is complete.
\end{pf}

When the condition in Theorem \ref{theorem_0_Sec4} is satisfied with positive definite weighting matrices, it follows from (\ref{eq:ProofTheorem0_0_Sec4}) that $v(\boldsymbol u_{k+1}) \le v(\boldsymbol u_{k+1},\boldsymbol r_{k+1})$ and $v(\boldsymbol r_{k+1}) \le v(\boldsymbol u_{k+1},\boldsymbol r_{k+1})$. This means the two-player end-to-end design can have a lower cost than any other one-player design according to the definition (\ref{eq:CharaFunc_Sec4}). Then, the grand coalition $\boldsymbol U$ is a dominating coalition compared to any other coalition. Note that both parallel and serial ILC designs are special cases of the end-to-end ILC, where only one player consists of the coalition. Therefore, the superiority of the proposed method is theoretically established in Theorem \ref{theorem_0_Sec4}.

\section{Case Study}\label{section V}
In this section, the developed end-to-end ILC is verified on a numerical model of a desktop printer with one translational degree of freedom as the one in \cite{vanMeer2022Gaussian}. Ideally, the printhead should move as fast as possible when equipped with a high-quality inkjet system as in a standard printer system. This motion problem is approximately equivalent to a step-response-shaped tracking task, which is untrackable in practice as defined in Definition \ref{def:trackable}. Research on how to plan a suitable trajectory has been conducted for better tracking performance, aiming to improve the efficiency of the printer. In this paper, the trackable trajectory will be automatically generated in the end-to-end ILC design.

The transfer function of the employed desktop printer is
\begin{equation}
	{\rm H}(s) \!=\! \frac{0.12s + 235}{9\times 10^{-5}s^4 + 1.092\times 10^{-2}s^3 + 21.385s^2},
\end{equation}
which is obtained by the system identification of a real desktop printer. This plant is stabilized in the time domain by the following feedback controller
\begin{equation}
	{\rm C}(s) = \frac{2.527\times 10^{5}s + 1.011\times 10^{7}} {s^2 + 351.9s + 6.317\times 10^{4}}.
\end{equation}
The sampling time is chosen as $10^{-3}{\rm s}$ in the discrete-time closed-loop control system with $4501$ samples in a trial.

The ILC objective is to track an untrackable impulse-response-shaped task, as the desired trajectory reference $\boldsymbol y_d$ plotted in Fig. \ref{fig:Output}. The end-to-end ILC design in Fig. \ref{fig:ControlDiagram} is employed to track $\boldsymbol y_d$. The weighting matrices in the cost function (\ref{eq:CostUnconstraint_Sec3}) are chosen to be symmetric and positive definite, i.e., $\boldsymbol Q = 10^{3}\cdot\boldsymbol I$, $\boldsymbol R = 10^{-2}\cdot\boldsymbol I$, $\boldsymbol S = 10^{-3}\cdot\boldsymbol I$, $\boldsymbol W = 10^{3}\cdot\boldsymbol I$, $\boldsymbol W_r = 10^{3}\cdot\boldsymbol I$, satisfying (\ref{Theorem0_Condition2_Sec4}) in Theorem \ref{theorem_0_Sec4}. Also, the initial ILC input is chosen as $\boldsymbol u_0 = \boldsymbol 0$.

\begin{figure}[tb]
	\centering
	\includegraphics[width=2.85in]{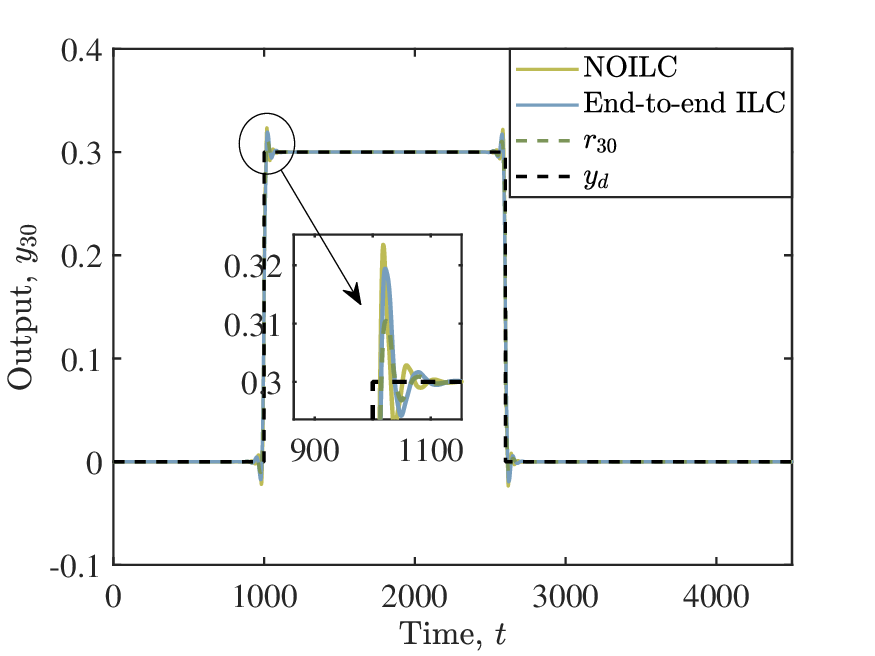}
	\caption{The $30$th tracking outputs of the end-to-end ILC compared to NOILC.}
	\label{fig:Output}
\end{figure}

\begin{figure}[tb]
	\centering
	\includegraphics[width=2.85in]{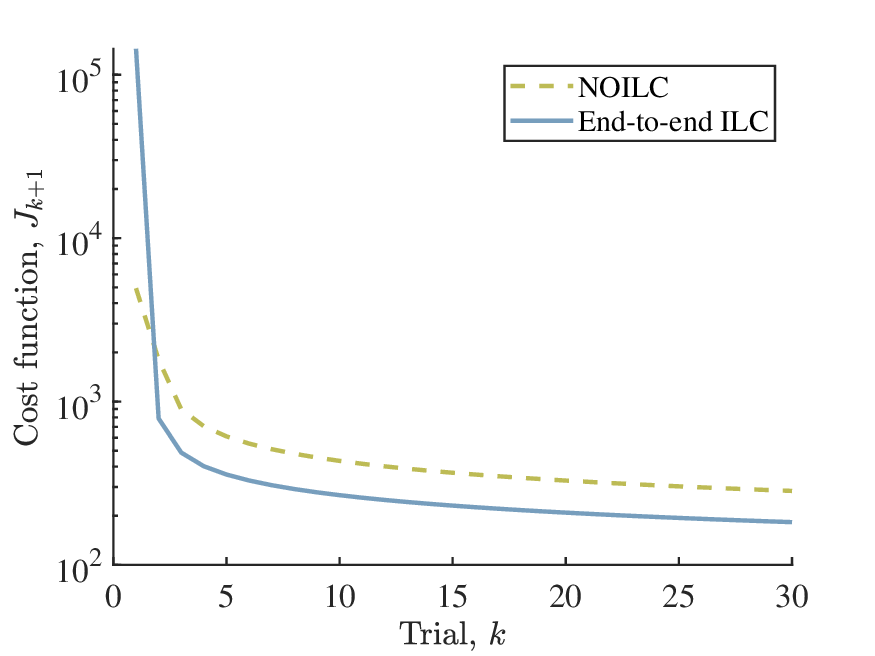}
	\caption{Cost convergence of the end-to-end ILC compared to NOILC.}
	\label{fig:Cost}
\end{figure}

Fig. \ref{fig:Output} shows the comparison results of the tracking output profiles after $30$ trials with the standard NOILC \citep{Amann1996_NOILC}, where NOILC chooses the same weighting parameters as the end-to-end ILC where $\boldsymbol r_{k+1} = \boldsymbol y_d$. The profile of the process trajectory $\boldsymbol r_{k+1}$ is also given in Fig. \ref{fig:Output}. It is shown that a trackable trajectory is learned with respect to the given cost function.

Moreover, the end-to-end ILC design can complete the task with better cost performance, which is demonstrated in the trial cost convergence in Fig. \ref{fig:Cost}. Note that the cost of the end-to-end ILC always has an additional $\boldsymbol W$-item compared to that of NOILC where $\boldsymbol r_{k+1} = \boldsymbol y_d$. Nevertheless, the lower cost of end-to-end ILC demonstrates that the cooperation between the two players of the closed-loop ILC system can better address the untrackable task.

\section{Conclusion and Future Work}\label{section VI}

In this paper, an end-to-end ILC design is developed in the closed-loop control systems for repetitive untrackable tasks. An adjustable trajectory is introduced to replace the considered untrackable trajectory reference in the closed-loop systems, which is updated together with the ILC feedforward input in the trial domain. From the cooperative game perspective, the adjustable process trajectory and the ILC input are two players of the control game for the closed-loop systems. A sufficient condition for the end-to-end ILC achieving a lower cost than NOILC is rigorously given based on the cooperative game theory. The result is verified on a simulation model of a desktop printer. For future work, the constrained case will be considered since an untrackable task usually means the violation of input restrictions in practice.

\begin{ack}
The first author would like to thank Max van Meer and Jizheng Liu for the fruitful discussions.
\end{ack}
%

\bibliography{ifacconf}             

\end{sloppypar}
\end{document}